\documentclass[onecolumn,journal]{IEEEtran}
\ifCLASSINFOpdf
\else
\fi

\usepackage[T1]{fontenc}
\usepackage[latin9]{inputenc}
\usepackage{color}
\usepackage{array}
\usepackage{float}
\usepackage{mathrsfs}
\usepackage{mathtools}
\usepackage{amsthm}
\usepackage{amstext}
\usepackage{amssymb}
\usepackage{stmaryrd}
\usepackage{graphicx}
\usepackage{pgfplots}

\makeatletter

\theoremstyle{plain}
\newtheorem{thm}{\protect\theoremname}
\theoremstyle{plain}
\newtheorem{prop}[thm]{\protect\propositionname}
\theoremstyle{plain}
\newtheorem{cor}[thm]{\protect\corollaryname}
\theoremstyle{plain}
\newtheorem{lem}[thm]{\protect\lemmaname}
\theoremstyle{definition}
\newtheorem{example}[thm]{\protect\examplename}
\theoremstyle{definition}
\newtheorem{defn}[thm]{\protect\definitionname}

\AtBeginDocument{
	
}

\makeatother

  \providecommand{\corollaryname}{Corollary}
  \providecommand{\examplename}{Example}
  \providecommand{\lemmaname}{Lemma}
  \providecommand{\propositionname}{Proposition}
  \providecommand{\theoremname}{Theorem}
  \providecommand{\definitionname}{Definition}

\newcommand{\res}{\operatorname{res}}

\begin{document}

%
\title{Multi-point Codes over Kummer Extensions}
%
%
%

\author{Chuangqiang~Hu and~Shudi~Yang    
\thanks{C. Hu is with the School of Mathematics, Sun Yat-sen University, Guangzhou 510275, P.R.China,\protect\\
\quad	S. Yang is with the School of Mathematical
	Sciences, Qufu Normal University, Shandong 273165, P.R.China. \protect\\
	\protect\\
	E-mail: huchq@mail2.sysu.edu.cn,~{yangshd3@mail2.sysu.edu.cn}}\protect\\
\thanks{Manuscript received *********; revised ********.}
}

\maketitle

\begin{abstract}
This paper is concerned with the construction of algebraic geometric codes defined from Kummer extensions. It plays a significant role in the study of such codes to describe bases for the Riemann-Roch spaces associated with totally ramified places.
Along this line, we give an explicit characterization of Weierstrass semigroups and pure gaps. Additionally, we determine the floor of a certain type of divisor introduced by Maharaj, Matthews and Pirsic. Finally, we apply these results to find multi-point codes with good parameters.  As one of the examples, a presented code with parameters $ [254,228,\geqslant 16] $ over $ \mathbb{F}_{64} $ yields a new record.

\end{abstract}

\begin{IEEEkeywords}
Algebraic geometric codes, Kummer extension, Weierstrass semigroup, Weierstrass pure gap.
\end{IEEEkeywords}

%
\IEEEpeerreviewmaketitle

\section{Introduction}
%
%
%
%
\IEEEPARstart{T}{he} algebraic geometric (AG) codes were introduced by V.D. Goppa~\cite{goppa1977codes}, which were defined as the image of the Riemann-Roch space by the evaluation at several rational places. Since then, the study of AG codes becomes an important instrument in theory and practice. The famous Tsfasman-Vl\v{a}du\c{t}-Zink theorem says
that the parameters of the AG codes associated with asymptotically good towers are better that the Gilbert-Varshamov bound~\cite{Tsfasman,niederreiter2001rational}. Pellikaan, Shen and van Wee~\cite{pellikan1991linear} showed that any arbitrary linear code is in fact an AG code.

Given an AG code of fixed length, the first task is to determine its parameters: dimension and minimum distance. In order to determine the dimension and construct the generator matrix, it is necessary to calculate the related Riemann-Roch space.
By means of the Riemann-Roch theorem, one obtains a non-trivial lower bound, Goppa bound, for the minimum distance in a very general setting~\cite{stichtenoth2009algebraic}.
Garcia, Kim and Lax improved the Goppa bound using arithmetical structure of the Weierstrass gaps at one place in~\cite{garcia1993consecutive,garcia1992goppa}.
 Homma and Kim~\cite{Homma2001Goppa} introduced the concept of pure gaps and demonstrated a similar result for a divisor concerning a pair of places. And this was generalized to several places by Carvalho and Torres in~\cite{carvalho2005goppa}.
Maharaj, Matthews and Pirsic~\cite{Maharaj2005riemann,maharaj2006floor} extended this construction by introducing the notion of the floor of a divisor and obtained improved bounds on the parameters of AG codes.

Codes over specific Kummer extensions were well-studied in the literature. For instance, Hermitian curves play an important role in coding theory due to their efficient encoding and decoding algorithms. Almost all of the known maximal curves arise from Hermitian curves. See~\cite{cossidente1999curves,matthews2005weierstrass} and the references therein. Many authors examined one-point codes from Hermitian curves and developed efficient methods to decode them~\cite{stichtenoth2009algebraic,Guruswami,Yang,Yang2}. The minimum distance of Hermitian two-point codes had been first determined by Homma and
Kim~\cite{Homma,Homma2,Homma3,Homma4}. In~\cite{Maharaj2005riemann}, Maharaj and Matthews determined explicit bases for the Riemann-Roch space of a divisor of the form $ rP_{\infty}+E $, where the support of $ E $ lies on a line. This allowed them to give an explicit formula for the floor of such a divisor. In~\cite{Geil2003normtrace}, Geil considered codes from norm-trace curves and determined the true minimum distance of these codes. Matthews~\cite{matthews2005weierstrass} determined the Weierstrass semigroup of
any $ r $-tuple rational points on the quotient of the Hermitian curve defined by the equation $ y^q + y = x^m $ over
$ \mathbb{F}_{q^2} $ where $  m>2 $ is a divisor of $ q + 1 $.
Sep\'{u}lveda and Tizziotti~\cite{sepulveda2014weierstrass} investigated two-point codes over a specific Kummer extension given by $ y^{q^l+1}=x^q+x $.

In this paper, we extend the results of one- and two-point codes over Kummer extensions studied by Masuda, Quoos and Sep{\'u}lveda~\cite{Masuda2}, to multi-point codes. We consider Kummer extensions given by $ y^m=f(x)^{\lambda} $ where $ f(x) $ is a polynomial over $ \mathbb{F}_q $ of degree $ r $ with
$ \gcd(m,r\lambda)=1 $, and all the roots of $ f(x) $ are pairwise distinct.
Let $ G $ be a divisor such that whose support is contained in one of the principal divisor of $ y $. An explicit basis for the Riemann-Roch space $ \mathcal{L}(G) $ is determined by constructing a related set of lattice points. Employing this result, we characterize the Weierstrass semigroups and the pure gaps with respect to several totally ramified places. In addition, we give an effective algorithm to compute the floor of the divisor $ G $. Finally, all these results lead us to find new codes with better parameters in comparison with the existing codes in the MinT's Tables~\cite{MinT}. A new record-giving $ [254,228,\geqslant 16] $-code over $ \mathbb{F}_{64} $ is presented as one of the examples.

The remaider of the paper is organized as follows. In Section~\ref{sec:Pre} we briefly recall some preliminary results over arbitrary function fields. Section~\ref{sec:Bases} focuses on the construction of bases for the Riemann-Roch space over Kummer extensions. In Section~\ref{sec:Weiersemipureg} we compute the Weierstrass semigroups and the pure gaps. Finally, in Section~\ref{sec:Examples} we construct multi-point codes with good parameters by employing our results.

\section{Preliminary results over arbitrary function fields}\label{sec:Pre}
Let $ q  $ be a power of a prime $p$ and $ \mathbb{F}_{q} $ be a finite field of cardinality $ q  $, with characteristic $ p $.  We denote by $ F $ a function field over $ \mathbb{F}_q$ and by
$ \mathbb{P}_F $ the set of places of $ F $. The free abelian group generated by the places of $ F $ is denoted by $ \mathcal{D}_F $, whose element is called a divisor. If a divisor $D$ is given by $ D=\sum_{ P \in {\mathbb{P}_F}} n_P P$ and almost all $ n_P=0 $, then the degree of $ D $ is $ \deg(D)= \sum_{P \in \mathbb{P}_F} n_P $.
For a function $ f \in F $, $ v_P(f) $ represents the valuation of $ f $ at a rational place $ P $. The divisor of $ f $ will be denoted by $ (f) $ and the divisor of poles of $ f $ will be denoted by $ (f)_{\infty} $.
 The Riemann-Roch vector space with respect to $ D $ is defined by
\[
\mathcal{L}(D)=\Big\{f \in F~\Big|~(f)+D \geqslant 0 \Big\} \cup \{0\}.
\]
 Let $ \ell(D) $ be the dimension of $ \mathcal{L}(D) $. From the Riemann-Roch Theorem, we know that
\[
\ell(D)-\ell(W-D)=1+g-\deg(D),
\]
 where $ W $ is the canonical divisor and $ g $ is the genus of the associated curve.

Let $ G $ be a divisor of $ F $ and let $ D:=Q_1+\cdots+Q_n $ be another divisor of $ F $ such that $ Q_1,\cdots,Q_n $ are distinct rational places, each not belonging to the support of $ G $. The AG codes
 $ C_{\mathcal{L}} $ and $ C_{\Omega} $ are defined as follows. The code
 $ C_{\mathcal{L}} $ is constructed from the Riemann-Roch space $ \mathcal{L}(G) $,
 \begin{equation*}
 C_{\mathcal{L}}:=\Big\{ (f(Q_1),\cdots,f(Q_n))~\Big|~f\in \mathcal{L}(G) \Big\} \subseteq \mathbb{F}_q^n.
 \end{equation*}
It is known that $ C_{\mathcal{L}} $ is an $ [n,k,d] $ code with parameters $ k=\ell(G)-\ell(G-D) $ and $ d \geqslant n-\deg(G) $. The code $ C_{\Omega} $
depends on the space of differentials $ \Omega(G-D) $,
\begin{equation*}
C_{\Omega}:= \Big\{ (\res_{Q_1}(\eta),\cdots,\res_{Q_n}(\eta))~\Big|~ \eta \in \Omega(G-D)  \Big\}.
\end{equation*}
Then $ C_{\Omega} $ is an $ [n,k_{\Omega},d_{\Omega}] $ code with parameters
$ k_{\Omega}=\ell(W+D-G)-\ell(W-G) $ and $ d_{\Omega}\geqslant \deg(G)-(2g-2) $. Under the hypothesis that $ \deg(G) > 2g-2 $, we have $ k_{\Omega}=\ell(W+D-G) \geqslant  n+g-1-\deg(G) $. If moreover $ 2g-2 < \deg(G) < n $ then \begin{equation}\label{eq:dimofCLk}
k_{\Omega}=n+g-1-\deg(G).
\end{equation}
The codes $ C_{\mathcal{L}} $ and $ C_{\Omega} $ are dual codes. We refer the reader to~\cite{stichtenoth2009algebraic} for more information.

We follow the notations in~\cite{matthews2004weierstrass}. Given $ l $ distinct rational places of $ F $, named $ Q_1,\cdots, Q_l $, the Weierstrass semigroup
$ H(Q_1,\cdots, Q_l) $ is defined by
\[
\Big\{(s_1,\cdots, s_l)\in \mathbb{N}_0^l~\Big|~\exists f\in F~ \text{with}~ (f)_{\infty}=\sum_{i=1}^l s_i Q_i  \Big\},
\]
and the Weierstrass gap set $  G(Q_1,\cdots, Q_l)  $ is defined by $ \mathbb{N}_0^l \backslash H(Q_1,\cdots, Q_l) $, where $ \mathbb{N}_0 := \mathbb{N}\cup \{0\} $ denotes the set of nonnegative integers.

An important subset of the Weierstrass gap set is the set of pure gaps.
Homma and Kim~\cite{Homma2001Goppa} introduced the pure gap set of two rational places. Carvalho and Torres~\cite{carvalho2005goppa} generalized this notion to several rational places, denoted by $ G_0(Q_1,\cdots, Q_l) $, which is given by
\begin{align*}
\Big\{&(s_1,\cdots,s_l)\in \mathbb{N}^l~\Big|~\ell(G) = \ell(G -Q_j ) ~\text{for}~1\leqslant j \leqslant l, \\
&~\text{where}~G=\sum_{i=1}^l s_iQ_i \Big\}.
\end{align*}
In addition, they showed that $ (s_1,\cdots,s_l)  $ is a pure gap at $ (Q_1,\cdots, Q_l) $ if and only if
\begin{align*}
\ell(s_1Q_1+\cdots+s_l Q_l)=\ell((s_1-1)Q_1+\cdots+(s_l-1) Q_l).
\end{align*}

The main motivation why one is interested in pure gap sets comes from constructing codes with excellent parameters, which will make use of the following theorem.


\begin{thm}[\cite{Homma2001Goppa,carvalho2005goppa}]\label{thm:Goppagener}
	Let $ Q_1 ,\cdots, Q_l $ be rational places of $ F $. For $ t_1,\cdots, t_l\in \mathbb{N} $, assume that
	\begin{align*}
	\Big\{(k_1,\cdots, k_l)\in \mathbb{N}^l ~\Big|~ \beta_1 \leqslant k_1 \leqslant \beta_1+t_1,
	\cdots, \beta_l \leqslant k_l \leqslant \beta_l+t_l   \Big\}
	\end{align*}
	is a subset of $  G_0(Q_1,\cdots, Q_l) $.
	If $ G=\sum_{i=1}^l(2\beta_i+t_i-1)Q_i $, then
	\begin{align*}
	d_{\Omega} \geqslant \deg(G)-(2g-2)+\sum_{i=1}^l t_i +l.
	\end{align*} 	
\end{thm}
The following lemma, which is an easy generalization of a result due to Kim~\cite{kim}, provides us with a way to calculate the Weierstrass semigroups.
\begin{lem}
	\label{lem:Weierstsemi}
	For rational places $ Q_1,\cdots,Q_l $ with
	$ 1 \leqslant l \leqslant r $, then $ H(Q_1,\cdots,Q_l) $ is given by
	\begin{align*}
	\Big\{&(s_1,\cdots,s_l)\in \mathbb{N}_0^l~\Big|~\ell(G) \neq \ell(G -Q_j ) ~\text{for}~1\leqslant j \leqslant l, \\
	&~\text{where}~G=\sum_{i=1}^l s_iQ_i \Big\}.
	\end{align*}

\end{lem}

\section{Bases for Riemann-Roch spaces over Kummer extensions}\label{sec:Bases}

Let $ m  \geqslant 2 $ and $ \gcd(p,m)=1 $. In this paper, we consider a Kummer extension $F_{\mathcal{K}}/{\mathbb{F}_q(x)}$ defined by
$
y^m=f(x)^{\lambda}=\prod_{i=1}^{r}(x-\alpha_i)^{\lambda}
$,
where $ \gcd(m,r \lambda)=1 $,
$ \alpha_i \in \mathbb{F}_q$ and the $ \alpha_i $'s are pairwise distinct for $ 1\leqslant i\leqslant r $. The function field $ F_{\mathcal{K}} $ has
 genus $ g= (r-1)(m-1)/2$. Let $ P_1,\cdots,P_r $ be the places of the rational function field $ F_{\mathcal{K}} $ associated to the zeros of $ x-\alpha_1,\cdots,x-\alpha_r $, respectively, and $ P_{\infty} $ be the unique place at infinity. It follows from~\cite{Stichtenoth} that they are totally ramified in this extension.

 The following proposition describes some principle divisors of a Kummer extension.
\begin{prop}\label{prop:divisor}
	Let $F_{\mathcal{K}}/{\mathbb{F}_q(x)}$ be a Kummer extension given by
	\begin{equation}\label{eq:Kumext1}
	y^m=f(x)^{\lambda}=\prod_{i=1}^{r}(x-\alpha_i)^{\lambda},
	\end{equation}
	 where $ \alpha_i \in \mathbb{F}_q $ and $ \text{gcd}(m, r\lambda) =1 $. Then we have the following divisors in $F$:
	\begin{enumerate}
		\item	$ (x-\alpha_i)=mP_i-mP_{\infty}$,  for $ 1\leqslant i\leqslant r $,
		\item $ (y)=\lambda P_1+\cdots +\lambda P_r-r\lambda P_{\infty} $,
		\item $(f(x))=\sum_{i=1}^r mP_i-rmP_{\infty}$.
	\end{enumerate}	
	
\end{prop}
Let $ G:=\sum_{\mu=1}^r s_{\mu}P_{\mu} + t P_{\infty} $. Maharaj~\cite{maharaj2004code} showed that  the Riemann-Roch space
$ \mathcal{L}(G) $ can be decomposed as a direct sum of Riemann-Roch
spaces of divisors of the projective line. For applications to computing pure gaps, we would like to give an explicit basis of $ \mathcal{L}(G) $, which consists of monomials of $ r $ elements. Actually, we generalize the result of~\cite{Maharaj2005riemann} concerning about the basis of Hermitian curves.

Since $ \text{gcd}(m, \lambda) =1 $, there exist integers $ A $ and  $ B $ such that $ A \lambda+Bm =1 $, and thus, if we denote $ z=y^A f(x)^B $, then we obtain
\begin{equation}\label{eq:divz}
 (z)=  P_1+\cdots + P_r-r P_{\infty} .
\end{equation}
Suppose that $ i,j_2,j_3,\cdots,j_r \in \mathbb{Z} $, we define
\begin{align}\label{eq:Eijr}
 E_{i,j_2,j_3,\cdots,j_r} := z^i (x-\alpha_2)^{j_2} (x-\alpha_3)^{j_3} \cdots (x-\alpha_r)^{j_r}.
\end{align}
By Proposition~\ref{prop:divisor} and Equation~\eqref{eq:divz}, one can compute the divisor of $ E_{i,j_2,j_3,\cdots,j_r} $:
\begin{align}
(E_{i,j_2,j_3,\cdots,j_r})=&iP_1 + (i+m j_2) P_2
+ \cdots +(i+m j_r)P_r \nonumber \\
& -\left( ri+m(j_2+\cdots+j_r) \right) P_{\infty} .\label{eq:E}
\end{align}

For later use, we denote by $ \lfloor x \rfloor $ the largest integer not greater than $ x $ and by $ \lceil x \rceil $ the smallest integer not less than $ x $. It is easy to show
that $ j = \left\lceil \dfrac{\alpha}{\beta} \right\rceil $ if and only if
$ 0 \leqslant \beta j -\alpha < \beta $, where $ \beta \in \mathbb{Z}^+ $ and $ \alpha \in \mathbb{Z} $.

Let us denote the lattice point set
\begin{align*}
\Omega_{s_1,\cdots,s_r, t} := \Big\{& (i,j_2,j_3,\cdots,j_r)
~\Big|~ ~i+s_1 \geqslant 0, \\
&~ 0 \leqslant i + mj_{\mu} + s_{\mu} < m ~~\text{for} ~~ \mu =2,\cdots, r,\\
&~ri+m(j_2+\cdots+j_r)   \leqslant t
\Big\},
\end{align*}
or equivalently,
\begin{align}
\Omega_{s_1,\cdots,s_r, t} := \Big\{& (i,j_2,j_3,\cdots,j_r)
~\Big|~i+s_1 \geqslant 0, \nonumber \\
&~ j_{\mu} = \left \lceil \frac{-i-s_{\mu}}{m} \right \rceil  ~~\text{for} ~~ \mu =2,\cdots, r,\nonumber \\
&~ri+m(j_2+\cdots+j_r)   \leqslant t
\Big\}.\label{eq:omegarceiljmu}
\end{align}

The following lemma is crucial for the proof of
our key result. However, the proof of this lemma is technical, and will be completed later.

\begin{lem}\label{thm:omega}
	The number of lattice points in $ \Omega_{s_1,\cdots,s_r, t} $ can be expressed as:
	\begin{equation*}
	\#\Omega_{s_1,\cdots,s_r, t}= 1-g+s_1+\cdots+s_r+t,
	\end{equation*}
	for $ s_1+\cdots+s_r+ t\geqslant(2r-1)m $.
\end{lem}

Now we can easily prove the main result of this section.
\begin{thm}\label{thm:basis1}
	Let $ G:=\sum_{\mu=1}^r s_{\mu}P_{\mu} + t P_{\infty} $. The elements $ E_{i,j_2,j_3,\cdots,j_r} $ with $ (i,j_2,j_3,\cdots,j_r) \in \Omega_{s_1,\cdots,s_r, t} $ form a basis for the Riemann-Roch space
	$ \mathcal{L}(G) $. Moreover, we have
	$ \ell(G) = \# \Omega_{s_1,\cdots,s_r, t}$.
\end{thm}
\begin{proof}
	Let  $ (i,j_2,j_3,\cdots,j_r) \in \Omega_{s_1,\cdots,s_r, t} $. It follows from the definition that $ E_{i,j_2,j_3,\cdots,j_r} \in \mathcal{L}(G) $, where $ G=\sum_{\mu=1}^r s_{\mu}P_{\mu} + t P_{\infty} $.
	From Equation~\eqref{eq:E}, we have
	$ v_{P_1}(E_{i,j_2,j_3,\cdots,j_r}) = i$, which indicates that the valuation of $E_{i,j_2,j_3,\cdots,j_r}$ at the rational place $P_1$ uniquely depends on $i$. Since lattice points in $ \Omega_{s_1,\cdots,s_r, t} $ provide distinct values of $i$, the elements $ E_{i,j_2,j_3,\cdots,j_r} $
	are linearly indepedent of each other, with $ (i,j_2,j_3,\cdots,j_r) \in \Omega_{s_1,\cdots,s_r, t} $. To show that they form a basis for $ \mathcal{L}(G) $, the only thing  is to prove that
	\[
	\ell (G)
	= \#\Omega_{s_1,\cdots,s_r, t}.
	\]	
	For the case of $ s_1 $ sufficiently large, it follows from the Riemann-Roch Theorem and Lemma~\ref{thm:omega} that
	\begin{align*}
	 \ell (G) & = 1-g + \deg(G)\\
	& = 1-g + s_1+\cdots+s_r+t\\
	&= \#\Omega_{s_1,\cdots,s_r, t}.
	\end{align*}
	And this implies that
	$ \mathcal{L}(G) $
	is spanned by elements $ E_{i,j_2,j_3,\cdots,j_r} $ with $ (i,j_2,j_3,\cdots,j_r) \in \Omega_{s_1,\cdots,s_r, t} $.
	
	For the general case, we choose $ s_1'> s_1 $ large enough and set $ G' := s_1' P_1 +\sum_{\mu=2}^r s_{\mu}P_{\mu} + t P_{\infty} $.
	From above argument, we know that the elements $ E_{i,j_2,j_3,\cdots,j_r} $ with $ (i,j_2,j_3,\cdots,j_r) \in \Omega_{s_1',\cdots,s_r, t} $ span the whole space of	$ \mathcal{L}(G') $.
	Remember that
	$ \mathcal{L}(G) $ is a linear subspace of $ \mathcal{L}(G') $, which can be written as
	\begin{equation}
	\mathcal{L}(G) = \Big\{ f \in \mathcal{L}(G') ~\Big|~v_{P_1}(f)\geqslant -s_1\Big\}.
	\end{equation}
	Thus, we choose $ f \in \mathcal{L}(G)  $, and suppose that
	\begin{equation*}
	f=\sum_{(i,j_2,\cdots,j_r) \in \Omega_{s_1',\cdots,s_r, t}} a_{i} E_{i,j_2,\cdots,j_r},
	\end{equation*}
	since $ f \in \mathcal{L}(G') $. The valuation of $f$ at $ P_1 $ is $ v_{P_1}(f)=\min_{a_i\neq 0} \{ i \}$. Then the inequality $ v_{P_1}(f)\geqslant -s_1 $ gives that, if $ a_i \neq 0 $, then $ i \geqslant -s_1 $. Equivalently, if $ i < -s_1 $, then $ a_i=0 $.  From the definition of $ \Omega_{s_1,\cdots,s_r, t} $ and $ \Omega_{s_1',\cdots,s_r, t} $, we get that
		\begin{equation*}
		 f=\sum_{(i,j_2,\cdots,j_r) \in \Omega_{s_1,\cdots,s_r, t} } a_{i} E_{i,j_2,\cdots,j_r}.
		 \end{equation*}
	Then the theorem follows.
\end{proof}

We now turn to prove Lemma~\ref{thm:omega} which requires  a series of results listed as follows.

\begin{lem}\label{lem:PsiR}
  Let $ g= (r-1)(m-1)/2$ and $ \gcd(r,m)=1 $. Let $ t \in \mathbb{Z} $. Consider the lattice point set
	\begin{equation*}
	\Psi(t) = \Big\{ (I,k)
~\Big|~0 \leqslant I < m , ~rI \leqslant t-m k, ~ k \geqslant 0\Big\}.
	\end{equation*}
	If $ t \geqslant rm $, then the number $ \#\Psi(t) $ of $ \Psi(t) $ verifies the formula
	\begin{equation*}
	\#\Psi(t)= 1-g + t.
	\end{equation*}
\end{lem}
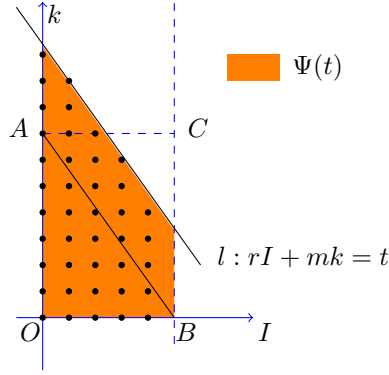
\begin{figure}[H]
	\centering
	\begin{tikzpicture}[scale=0.35]
	\path [fill=orange] (0,0)--(5,0) -- (5,3.32) -- (0,10.32);
	\draw [blue,->] (-1,0)--(8,0);
	\draw [blue,->] (0,-2)--(0, 12);
    \draw [blue,dashed] (5,-1)--(5,12);
    \draw [blue,dashed] (0,7)--(5,7);

    \draw [black] (5,0)--(0, 7);
    \draw [black] (6,2.6-0.6)--(-1,12.4-0.6);

    \draw[color=black] node [yshift=-1ex,xshift=1ex] at (0,12) { $ k $ };
    \draw[color=black] node [yshift=-1.2ex,xshift=1ex] at (8,0) { $ I $ };
    \draw[color=black] node [yshift=0.5ex,xshift=-2ex] at (0,7) { $ A $ };
    \draw[color=black] node [yshift=-1.2ex,xshift=-1ex] at (0,0) { $ O $ };
    \draw[color=black] node [yshift=-1.2ex,xshift=1ex] at (5,0) { $ B $ };
    \draw[color=black] node [yshift=0.5ex,xshift=2ex] at (5,7) { $ C $ };

    \draw[color=black] node [yshift=-1.51ex,xshift=2ex] at (9,3) { $ l:rI+mk=t $ };

    \draw[color=black] node [yshift=0ex,xshift=1ex] at (10,9.5) { $ \Psi(t) $ };

    \path [fill=orange] (7,9)--(7,10) -- (9,10) -- (9,9);

%
%
	
	\draw [fill] (0, 0) circle [radius=0.1];
	\draw [fill] (0, 1) circle [radius=0.1];
	\draw [fill] (0, 2) circle [radius=0.1];
	\draw [fill] (0, 3) circle [radius=0.1];
	\draw [fill] (0, 4) circle [radius=0.1];
	\draw [fill] (0, 5) circle [radius=0.1];
	\draw [fill] (0, 6) circle [radius=0.1];
	\draw [fill] (0, 7) circle [radius=0.1];
	\draw [fill] (0, 8) circle [radius=0.1];
	\draw [fill] (0, 9) circle [radius=0.1];
	\draw [fill] (0, 10) circle [radius=0.1];
	\draw [fill] (1, 0) circle [radius=0.1];
    \draw [fill] (1, 1) circle [radius=0.1];
    \draw [fill] (1, 2) circle [radius=0.1];
    \draw [fill] (1, 3) circle [radius=0.1];	
    \draw [fill] (1, 4) circle [radius=0.1];
    \draw [fill] (1, 5) circle [radius=0.1];
    \draw [fill] (1, 6) circle [radius=0.1];
    \draw [fill] (1, 7) circle [radius=0.1];
    \draw [fill] (1, 8) circle [radius=0.1];
    \draw [fill] (1, 9) circle [radius=0.1];
    \draw [fill] (2, 0) circle [radius=0.1];
    \draw [fill] (2, 1) circle [radius=0.1];
    \draw [fill] (2, 2) circle [radius=0.1];
    \draw [fill] (2, 3) circle [radius=0.1];
    \draw [fill] (2, 4) circle [radius=0.1];
    \draw [fill] (2, 5) circle [radius=0.1];
    \draw [fill] (2, 6) circle [radius=0.1]; 	
	\draw [fill] (2, 7) circle [radius=0.1];
    \draw [fill] (3, 0) circle [radius=0.1];
    \draw [fill] (3, 1) circle [radius=0.1];
    \draw [fill] (3, 2) circle [radius=0.1];
    \draw [fill] (3, 3) circle [radius=0.1];
    \draw [fill] (3, 4) circle [radius=0.1];
    \draw [fill] (3, 5) circle [radius=0.1];
    \draw [fill] (3, 6) circle [radius=0.1];
    \draw [fill] (4, 0) circle [radius=0.1];
    \draw [fill] (4, 1) circle [radius=0.1];
    \draw [fill] (4, 2) circle [radius=0.1];
    \draw [fill] (4, 3) circle [radius=0.1];
    \draw [fill] (4, 4) circle [radius=0.1];
	\end{tikzpicture}
	
	\protect\caption{The lattice point set $\Psi(t) $ }
	\label{fig:Psi_t}
\end{figure}
\begin{proof}
	Let $ t_0 = rm $. As shown in Fig.\ref{fig:Psi_t}, we denote by $ \# \triangle OAB $ the number of lattice points in the triangle $ \triangle OAB $ including the edges $ OA $, $ OB $ and $ AB $, with $ O=(0,0) $, $ A=(0, r) $ and $ B= (m,0) $. It is easy to see that the number of lattice points in the rectangle $ \# \square AOBC $ (including all the edges) satisfies $ \# \square AOBC =(m+1)(r+1)$, where $ C=(m,r) $. Clearly $  \# \triangle OAB + \# \triangle ABC = \# \square AOBC +2 $. This indicates that
	$ \# \triangle OAB = (m+1)(r+1)/2 +1$.
	But $ \Psi(t_0) $ contains exactly the lattice points in the triangle $ \triangle OAB $ except the vertex $ B $, so
	\begin{align*}
	\#\Psi(t_0)&= \# \triangle OAB-1\\
	&=\dfrac{(m+1)(r+1)}{2}\\
	&=1-g+t_0.	 		
	\end{align*}	 	
	For $ t > t_0 $, consider the lattice points in the set $ \Psi(t) \backslash \Psi(t-1) $, which can be represented by
	\begin{equation*}
     \Delta := \Big\{ (I,k)~\Big|~0 \leqslant I < m , ~rI+mk=t\Big\}.
	\end{equation*}	
 	The equation $ rI+mk=t $ has integer solutions since
$ \gcd(r,m)=1 $ and we assume that $ (I_0, k_0) $ is such a solution. Notice that all the other solutions are given by
	$ (I_0 +m \gamma, k_0 - r\gamma) $ with $ \gamma \in \mathbb{Z} $. This implies that
		\begin{equation*}
		\Delta=\Big\{ (I_0 +m \gamma, k_0 - r\gamma)~\Big|~0 \leqslant I_0 +m \gamma < m \Big\},
		\end{equation*}
		which gives that $ \# \Delta =1 $. Thus we conclude that
		\begin{equation*}
		\#\Psi(t)= t-t_0+\#\Psi(t_0)=1-g+t.
		\end{equation*}	
\end{proof}

From Lemma~\ref{lem:PsiR}, we obtain the number of lattice points in $ \Omega_{0,\cdots,0, t} $.
\begin{lem}\label{lem:omega0000t}
	If $ t\geqslant rm $, then the number of lattice points in $ \Omega_{0,\cdots,0, t} $ is
	\begin{equation}
	\#\Omega_{0,\cdots,0, t}= 1-g+t.
	\end{equation}
\end{lem}

\begin{proof}
	Note that
	\begin{align}\label{eq:omega000t}
	\Omega_{0,\cdots,0, t} =\Big \{& (i,j_2,j_3,\cdots,j_r)
	~\Big|~i \geqslant 0, \nonumber \\
	&~ 0 \leqslant i + mj_{\mu}  < m ~~\text{for} ~~ \mu =2,\cdots, r, \nonumber \\
	&~ri+m(j_2+\cdots+j_r)   \leqslant t
	\Big\},
	\end{align}
	Let $ I := i+ m j_2 $. We will show that the number of lattice points in $ \Omega_{0,\cdots,0, t} $ equals
	\begin{align}\label{eq:omega000t2}
	\# \Big\{ (I,k)
	~\Big|~0 \leqslant I < m , ~rI \leqslant t-m k, ~ k \geqslant 0
	\Big\}.
	\end{align}
	Then the assertion follows at once from Lemma~\ref{lem:PsiR}.
	
	Set  $ J_{\mu} := j_{\mu} - j_2 $, for $ \mu \geqslant 3 $.
  Then Equation~\eqref{eq:omega000t} gives that
	\begin{align*}
	\Omega_{0,\cdots,0, t} \cong  \Big\{& (I,j_2,J_3,\cdots,J_r)
	~\Big|~I \geqslant  m j_2, \nonumber \\
	& 0 \leqslant I < m, \nonumber \\
	& 0 \leqslant I+m J_{\mu}  < m, ~~\text{for} ~~ \mu \geqslant 3, \nonumber \\
	&rI-m j_2+m\sum_{\mu=3}^r J_{\mu} \leqslant t
	\Big\},\label{eq:omega0001t}
	\end{align*}
	Here and thereafter,
	 $ A \cong B $ means that two lattice point sets $ A $, $ B $ are bijective.
	From $0 \leqslant I < m $ and $ I \geqslant  m j_2 $, we must have $ j_2 \leqslant 0 $. Hence
	\begin{align*}
	\Omega_{0,\cdots,0, t} \cong \Big\{& (I,j_2,J_3,\cdots,J_r)
	~\Big|~j_2 \leqslant 0,  ~0 \leqslant I < m, \nonumber \\
	&~ 0 \leqslant I+m J_{\mu} < m, ~~\text{for} ~~ \mu \geqslant 3, \nonumber \\
	&~rI-m j_2+m\sum_{\mu=3}^r J_{\mu} \leqslant t
	\Big\}.
	\end{align*}
	However, $ 0 \leqslant I < m $ gives that $ J_{\mu}=0 $ for $ \mu \geqslant 3$. This implies that
	\begin{align*}
	\Omega_{0,\cdots,0, t} \cong \Big\{ (I,j_2)
	~\Big|~j_2 \leqslant 0,  ~0 \leqslant I < m,~rI-m j_2 \leqslant t   	\Big\}.
	\end{align*}
	Then Equation~\eqref{eq:omega000t2} follows from the above discussions and the assumption $ k = - j_2 $.
\end{proof}

\begin{lem}\label{lem:omegas2=0}
	Let $ s_2=0 $. The following equality holds:
	\[
	\# \Omega_{s_1,0, s_3, \cdots,s_r, t} = 1-g+t+\sum_{i=1}^r s_i ,
	\]
	where $ 1\leqslant s_i \leqslant m $ for $ i=1,3,4,\cdots, r $, and $ t \geqslant rm $.
\end{lem}
\begin{proof}
	It follows from the definition that
	\begin{align*}
	\Omega_{s_1,0, s_3, \cdots,s_r, t} = \Big\{& (i,j_2,j_3,\cdots,j_r)
	~\Big|~i+s_1 \geqslant 0, \\
	&~ 0 \leqslant i + mj_2 < m ,\\
	&~ 0 \leqslant i + mj_{\mu} + s_{\mu} < m ~~\text{for} ~~ \mu \geqslant 3,\\
	&~ri+m(j_2+\cdots+j_r)   \leqslant t
	\Big\}.
	\end{align*}
	Put $ I := i+ m j_2 $, $ J_{\mu} := j_{\mu} - j_2 $ for $ \mu \geqslant 3 $.  Then
		\begin{align}\label{eq:omegawithj2}
		\Omega_{s_1,0, s_3, \cdots,s_r, t} \cong \Big\{& (I,j_2,J_3,\cdots,J_r)
		~\Big|~I \geqslant mj_2 -s_1, \nonumber \\
		&~ 0 \leqslant I < m ,\nonumber \\
		&~ 0 \leqslant I+ mJ_{\mu} + s_{\mu} < m ~~\text{for} ~~ \mu \geqslant 3,\nonumber \\
		&~rI-mj_2+m(J_3+\cdots+J_r)   \leqslant t
		\Big\}.
		\end{align}
	It turns out that $j_2 \leqslant 1$, since $ 1\leqslant s_1 \leqslant m $, $ I \geqslant mj_2 -s_1 $ and $ 0 \leqslant I < m $.
	Hence
			\begin{align*}
	 \#\Omega_{s_1,0, s_3, \cdots,s_r, t} =
	 \sum_{j_2=-\infty}	^1 \# \Psi_{j_2},
		\end{align*}
	where 	$ \Psi_{j_2} $ denotes the right hand side of Equation~\eqref{eq:omegawithj2} with a fixed $ j_2 $. We first calculate $ \#\Psi_1 $, where
	\begin{align*}
	\Psi_{1}  = \Big\{& (I,J_3,\cdots,J_r)
	~\Big|~ m-s_1 \leqslant I < m ,\\
	 &~ J_{\mu} = \left \lceil \frac{-I-s_{\mu}}{m} \right \rceil  ~~\text{for} ~~ \mu \geqslant 3,\\
	&~rI+m(J_3+\cdots+J_r)   \leqslant t+m
	\Big\}.
	\end{align*}
	Since $ J_{\mu} \in \{0,-1\} $ for $ \mu \geqslant 3 $, and $ m(J_3+\cdots+J_r) \leqslant 0 $,
	if we choose  $ C := r(m-1)-m $, then the last inequality in $ \Psi_{1} $ holds for all $ t \geqslant C $.  This shows that $ \#\Psi_1 $ is determined by the first inequality. Then
	$ \#\Psi_1 =s_1$. Therefore
	\begin{align}\label{eq:omegas2=0}
	\# \Omega_{s_1,0, s_3, \cdots,s_r, t} =
	\sum_{j_2=-\infty}	^0 \# \Psi_{j_2}+s_1.
	\end{align}
	
	Next, we turn to calculate the total number of lattice points in $ \Psi_{j_2} $ for all $ j_2 \leqslant 0 $. For this purpose, we write
	\begin{align}\label{eq:psij2andj2}
	 \#\Psi_{j_2} = \#\Psi_{j_2}' +\# \Psi_{j_2}''  ,
	 \end{align} where
	\begin{align*}
	\Psi_{j_2}' : = \Big\{& (I,J_3,\cdots,J_r)~\Big|~ 0 \leqslant I < m ,\\
	&~ -s_3 \leqslant I+mJ_3<0,\\
	&~ 0 \leqslant I+ mJ_{\mu} + s_{\mu} < m ~~\text{for} ~~ \mu \geqslant 4,\\
	&~rI+m(J_3+\cdots+J_r)   \leqslant t + mj_2
	\Big\},
	\end{align*}
   and
   	\begin{align*}
   	\Psi_{j_2}'' : = \Big\{& (I,J_3,\cdots,J_r)~\Big|~ 0 \leqslant I < m ,\\
   	&~ 0 \leqslant I+mJ_3< m-s_3,\\
   	&~ 0 \leqslant I+ mJ_{\mu} + s_{\mu} < m ~~\text{for} ~~ \mu \geqslant 4,\\
   	&~rI+m(J_3+\cdots+J_r)   \leqslant t + mj_2
   	\Big\}.
   	\end{align*}
   	If we let $ J_3':=J_3+1 $, then the following assertion holds:
   		\begin{align*}
   		\Psi_{j_2}'  \cong \Big\{& (I,J_3',\cdots,J_r)~\Big|~ 0 \leqslant I < m ,\\
   		&~ m-s_3 \leqslant I+mJ_3'<m,\\
   		&~ 0 \leqslant I+ mJ_{\mu} + s_{\mu} < m ~~\text{for} ~~ \mu \geqslant 4,\\
   		&~rI+m(J_3'+\cdots+J_r)   \leqslant t + m(j_2+1)
   		\Big\}.
   		\end{align*}
   		Observe that $ \#\Psi_{j_2-1}' + \# \Psi_{j_2}''$ is equivalent to $ \#\Phi_{j_2}  $, where
   		  		\begin{align*}
   		  		\Phi_{j_2} := \Big\{& (I,J_3,\cdots,J_r)~\Big|~ 0 \leqslant I < m ,\\
   		  		&~ 0 \leqslant I+mJ_3<m,\\
   		  		&~ 0 \leqslant I+ mJ_{\mu} + s_{\mu} < m ~~\text{for} ~~ \mu \geqslant 4,\\
   		  		&~rI+m(J_3+\cdots+J_r)   \leqslant t + mj_2
   		  		\Big\},
   		  		\end{align*}
   		  		and since $ \# \Omega_{0,0, 0,s_4 \cdots,s_r, t}=\sum_{j_2=-\infty}^0 \# \Phi_{j_2} $, we have
   		  			\begin{align}\label{eq:omegapsij2}
   		  			\# \Omega_{0,0, 0,s_4 \cdots,s_r, t}&= \sum_{j_2=-\infty}^0(\#\Psi_{j_2-1}' + \# \Psi_{j_2}'') \nonumber \\
   		  			&= \sum_{j_2=-\infty}^0(\#\Psi_{j_2}' + \# \Psi_{j_2}'')-\#\Psi_{0}'.
   		  			\end{align}
   	   	Precisely speaking,
   	   	\begin{align*}
   	   	\Psi_{0}' \cong \Big\{& (I,J_3,\cdots,J_r)~\Big|~ 0 \leqslant I < m ,\\
   	   	&~ -s_3 \leqslant I+mJ_3<0,\\
   	   	&~ 0 \leqslant I+ mJ_{\mu} + s_{\mu} < m ~~\text{for} ~~ \mu \geqslant 4,\\
   	   	&~rI+m(J_3+\cdots+J_r)   \leqslant t
   	   	\Big\}.   	   	
   	   	\end{align*}
   	   	In a similar way as we calculate $ \#\Psi_1 $, we can see that
   	   	$ \#\Psi_{0}'= s_3 $. This, together with Equations~\eqref{eq:omegas2=0},~\eqref{eq:psij2andj2} and~\eqref{eq:omegapsij2}, shows that
   	   	\begin{align*}
   	    \# \Omega_{s_1,0, s_3, \cdots,s_r, t}= \# \Omega_{0,0, 0,s_4 \cdots,s_r, t} + s_1 +s_3.
   	   	\end{align*}
   	   	Repeating the above routine gives that
   	   	   	   	\begin{align*}
   	   	   	   	\# \Omega_{s_1,0, s_3, \cdots,s_r, t}= \# \Omega_{0,0,  \cdots,0, t} + s_1 +s_3+\cdots+s_r.
   	   	   	   	\end{align*}
   	   	The desired conclusion then follows from Lemma~\ref{lem:omega0000t}.
\end{proof}
Now we are in a position to give the proof of Lemma~\ref{thm:omega}.
\begin{proof}[Proof of Lemma~\ref{thm:omega}]

Let us consider the lattice point set $ \Omega_{s_1,s_2,\cdots,s_r, t} $, which is given by
\begin{align*}
\Omega_{s_1,s_2,\cdots,s_r, t} = \Big\{& (i,j_2,j_3,\cdots,j_r)
~\Big|~i+s_1 \geqslant 0, \\
&~ 0 \leqslant i + mj_{2} + s_{2} < m, \\
&~ 0 \leqslant i + mj_{\mu} + s_{\mu} < m ~~\text{for} ~~ \mu \geqslant 3,\\
&~ri+m(j_2+\cdots+j_r)   \leqslant t
\Big\}.
\end{align*}
Let $ \tilde{i}:=i+s_2 $,  $ \tilde{t}: = t + rs_2 $ and $ \tilde{s}_{\mu}:=  s_{\mu} -s_2$ for $ \mu \geqslant 1$. Then $  \Omega_{s_1,s_2,\cdots,s_r, t} $ becomes
\begin{align*}
\Omega_{\tilde{s}_1,0,\tilde{s}_3,\cdots,\tilde{s}_r, \tilde{t}} = \Big\{& (\tilde{i},j_2,j_3,\cdots,j_r)
~\Big|~\tilde{i}+\tilde{s}_1  \geqslant 0, \\
&~ 0 \leqslant \tilde{i} + mj_{2}  < m, \\
&~ 0 \leqslant \tilde{i} + mj_{\mu} + \tilde{s}_{\mu} < m ~~\text{for} ~~ \mu \geqslant 3,\\
&~r\tilde{i}+m(j_2+\cdots+j_r)  \leqslant \tilde{t} ~
\Big\}.
\end{align*}
By writing $ \tilde{s}_{\mu}$ as $ \tilde{s}_{\mu} =  mA_{\mu} +B_{\mu}$, where $ 1\leqslant B_{\mu} \leqslant m $  for $ \mu \geqslant 1$, and taking $ I := \tilde{i}+mA_1 $, $ J_{\mu} := j_{\mu}-A_1+A_{\mu} $ for $\mu \geqslant 3$, $J_2 := j_{2}-A_1 $ and $ T :=\tilde{t}+m(A_1+\sum_{\mu=3}^r A_{\mu}) $, we can reduce  $ \Omega_{\tilde{s}_1,0,\tilde{s}_3,\cdots,\tilde{s}_r, \tilde{t}} $ to	
	\begin{align*}
	\Omega_{B_1,0,B_3,\cdots,B_r, T} = \Big\{& (I,J_2,J_3,\cdots,J_r)
	~\Big|~I+B_1  \geqslant 0, \\
	&~ 0 \leqslant I + mJ_{2}  < m, \\
	&~ 0 \leqslant I + mJ_{\mu} + B_{\mu} < m ~~\text{for} ~~ \mu \geqslant 3,\\
	&~rI+m(J_2+\cdots+J_r)  \leqslant T ~
	\Big\}.
	\end{align*}
	 Since the inequality $ s_1+\cdots+s_r+ t\geqslant(2r-1)m $ means that $ T \geqslant rm $, it follows from Lemma~\ref{lem:omegas2=0} that
\begin{align*}
\#\Omega_{B_1,0,B_3,\cdots,B_r, T}
& = 1-g+ T + \sum_{\mu=1 \atop \mu \neq 2}^r B_{\mu} \\
&=  1- g + \tilde{t}+\sum_{\mu=1 \atop
	\mu \neq 2}^r \tilde{s}_{\mu}\\
&= 1 -g + t + \sum_{\mu=1 }^r s_{\mu}.
\end{align*}
Note that
\begin{align*}
\#\Omega_{s_1,s_2,\cdots,s_r, t}=\#\Omega_{B_1,0,B_3,\cdots,B_r, T}.
\end{align*}
This finishes the proof of Lemma~\ref{thm:omega}.
	
\end{proof}

We finish this section with a result that allows us to give a new form of the base for our Riemann-Roch spaces. Since $ \gcd(r,m)=1 $, we write $ ar+ bm=1 $ for integers $a$ and $b$.  Denote
\begin{equation*}
\Lambda_{u,v_2,v_3,\cdots,v_r}:=\beta^u\prod_{\mu=2}^r h_{\mu}^{v_{\mu}},
\end{equation*}
 where $ h_{\mu}: = \dfrac{x-\alpha_{\mu}}{x-\alpha_1} $ for $ \mu \geqslant 2 $, and $ \beta:= z^{-a}(x-\alpha_1)^{-b} $. The divisor of $ \Lambda_{u,v_2,v_3,\cdots,v_r}$ is given by
\begin{align*}
(\Lambda_{u,v_2,\cdots,v_r})=
& (- (a+bm)u- m(v_2+\cdots+v_r))P_1 \\
&     + \sum_{\mu=2}^{r}(- au + mv_{\mu}) P_{\mu}+u P_{\infty}.
\end{align*}
\begin{cor}\label{cor:thetabase2}
		Let $ G:=\sum_{\mu=1}^r s_{\mu}P_{\mu} + t P_{\infty} $. Then the elements $ \Lambda_{u,v_2,\cdots,v_r} $ with $ (u,v_2,v_3,\cdots,v_r) \in \Theta_{s_1,\cdots,s_r, t} $ form a basis for the Riemann-Roch space
		$ \mathcal{L}(G) $, where
			\begin{align}
		\Theta_{s_1,\cdots,s_r, t} := \Big\{& (u,v_2,\cdots,v_r)~\Big| \nonumber \\
			 & - (a+bm)u-m(v_2+\cdots+v_r) +s_1 \geqslant 0  , \nonumber \\
			&~ 0 \leqslant - au + mv_{\mu} +s_{\mu} < m ~~\text{for} ~~ \mu \geqslant 2, \nonumber \\
			&~ u   \geqslant - t
			\Big\}.
			\end{align} \label{eq:thetaaaa}
			And then $ \#\Theta_{s_1,\cdots,s_r, t}=\# \Omega_{s_1,\cdots,s_r, t}  $.
\end{cor}
\begin{proof}
		Note that $ y^m=f(x)^{\lambda} $ and $ A\lambda+ Bm=1 $. Then
		$ y^{Am}=f(x)^{A\lambda} $, which gives that $ f(x)=y^{Am}f(x)^{Bm}=z^m $, where $ z=y^{A}f(x)^{B} $.
		From Equation~\eqref{eq:Kumext1}, we get
		\[
		x-\alpha_1 =z^m(x-\alpha_2)^{-1}\cdots(x-\alpha_{r})^{-1} .
		\]
	We claim that the set
	\begin{equation*}
	\Big\{ \Lambda_{u,v_2,\cdots,v_r}~\Big|~(u,v_2,\cdots,v_r)\in \Theta_{s_1,\cdots,s_r, t} \Big\}
	\end{equation*} equals the set
	\begin{equation*}
	\Big\{ E_{i,j_2,\cdots,j_r} ~\Big|~ (i,j_2,\cdots,j_r)\in \Omega_{s_1,\cdots,s_r, t}\Big\}.
	\end{equation*}
In fact, for fixed $(u,v_2,v_3,\cdots,v_r)\in \mathbb{Z}^r$, we obtain $\Lambda_{u,v_2,\cdots,v_r}$ equals $ E_{i,j_2,\cdots,j_r} $ with
\begin{align*}
i&= - (a+bm)u-m(v_2+\cdots+v_r) , \\
 j_{\mu}&=bu+(v_2+\cdots+v_r)+v_{\mu},
\end{align*}
 for $\mu \geqslant 2$. On the contrary, if we set
\begin{align*}
u&=-ri-m(j_2+\cdots+j_r) ,\\
 v_{\mu}&=bi-a(j_2+\cdots+j_r)+j_{\mu} ,
 \end{align*}
 for $\mu \geqslant 2$, then $ E_{i,j_2,\cdots,j_r} $ is exactly the element $ \Lambda_{u,v_2,\cdots,v_r} $.
Therefore, if we restrict $ (i,j_2,\cdots,j_r) $ in $ \Omega_{s_1,\cdots,s_r, t} $, then we must have  $(u,v_2,\cdots,v_r)$ is in $ \Theta_{s_1,\cdots,s_r, t} $ and vice versa. This completes the proof of the claim and hence of this corollary.
\end{proof}

\section{Weierstrass semigroups and pure gap sets}\label{sec:Weiersemipureg}
In this section, we calculate the Weierstrass semigroups and the pure gap sets at the totally ramified places $ P_1,\cdots, P_l , P_{\infty} $, which will require auxiliary results described below.


\begin{lem}\label{lem:omeganoorder}
The lattice point set $ \Omega_{s_1,\cdots, s_r, t} $ is symmetric with respect to $ s_1,\cdots, s_r $, which means that	$\# \Omega_{s_1,\cdots, s_r, t} = \# \Omega_{s_1',\cdots, s_r', t}  $, where $ \{s_1,\cdots, s_r\}=\{s_1',\cdots, s_r'\}   $.
\end{lem}
\begin{proof}
	Recall that $ \Omega_{s_1,\cdots,s_r, t} $  is defined by
	\begin{align*}
	\Omega_{s_1,\cdots,s_r, t} = \Big\{& (i',j_2',j_3',\cdots,j_r')
	~\Big|~i'+s_1 \geqslant 0, \nonumber \\
	&~ j_{\mu}' = \left \lceil \frac{-i'-s_{\mu}}{m} \right \rceil  ~~\text{for} ~~ \mu =2,\cdots, r,\nonumber \\
	&~ri'+m(j_2'+\cdots+j_r')   \leqslant t
	\Big\}.
	\end{align*}
It is important to write $ i'=i+mk $ with $ 0 \leqslant i \leqslant m-1 $. Let $  j_{\mu}' =  j_{\mu}-k $ for $  \mu \geqslant 2 $. Then
	\begin{align*}
	\Omega_{s_1,\cdots,s_r, t} \cong\Big\{& (i,k,j_2,j_3,\cdots,j_r)
	~\Big|~i+mk \geqslant -s_1,\\
	& ~0 \leqslant i \leqslant m-1, \nonumber \\
	&~ j_{\mu} = \left \lceil \frac{-i-s_{\mu}}{m} \right \rceil  ~~\text{for} ~~ \mu =2,\cdots, r,\nonumber \\
	&~ri+m(k+j_2+\cdots+j_r)   \leqslant t
	\Big\}.
	\end{align*}
The first inequality gives that  $ k\geqslant j_{1}: = \left \lceil \frac{-i-s_{1}}{m} \right \rceil  $. So we write $ k= j_{1} +\iota $ with $ \iota \geqslant 0 $.
 Then
		\begin{align*}
		\Omega_{s_1,\cdots,s_r, t}\cong \Big\{& (i,\iota, j_1,j_2,j_3,\cdots,j_r)
		~\Big|\\
		&  ~0 \leqslant i \leqslant m-1,~ \iota \geqslant 0,\\
		&~ j_{\mu} = \left \lceil \frac{-i-s_{\mu}}{m} \right \rceil  ~~\text{for} ~~ \mu =1,\cdots, r,\nonumber \\
		&~ri+m(j_1+j_2+\cdots+j_r) +m\iota  \leqslant t
		\Big\}.
		\end{align*}
The right hand side means that the number of the lattice points does not depend on the order of $s_{\mu}$ with $ 1 \leqslant \mu \leqslant r $, which concludes the desired assertion.		
\end{proof}

\begin{lem}\label{lem:omegawith1}
	Let $ \Omega_{s_1,s_2, \cdots,s_r, t} $ be the lattice point set defined by Equation~\eqref{eq:omegarceiljmu}. The following assertions hold.
	\begin{enumerate}
		\item
		$ \#	\Omega_{s_1,s_2, \cdots,s_r, t}  = \#	\Omega_{s_1-1,s_2,\cdots,s_r, t}+1  $ if and only if
		$ m \sum\limits_{\mu=2 }^r\left\lceil \dfrac{s_1-s_{\mu}}{m}\right\rceil
		\leqslant t+ rs_1  $.
				\item
				$ \#	\Omega_{s_1, \cdots,s_r, t}  = \#	\Omega_{s_1,\cdots,s_r, t-1}+1  $ if and only if
				$
				~ m \sum\limits_{\mu=2 }^{r}\left\lceil \dfrac{-at-s_{\mu}}{m}\right\rceil
				\leqslant s_1 + (a+bm)t $.
	\end{enumerate}
\end{lem}

\begin{proof}
	We begin with the first assertion.
	 Consider two lattice point sets $\Omega_{s_1,s_2, \cdots,s_r, t}$ and $\Omega_{s_1-1,s_2, \cdots,s_r, t}$, which are given in Equation~\eqref{eq:omegarceiljmu}. Clearly, the latter one is a subset of the former one, and the complementary set $ \Phi $ of $\Omega_{s_1-1,s_2, \cdots,s_r, t}$ in $\Omega_{s_1,s_2, \cdots,s_r, t}$ is given  by
	 	\begin{align*}\label{eq:omegarceiljmus}
	 	\Phi :
	 	 	 =&\Big\{ (i,j_2,\cdots,j_r) ~	\Big	|	~i+s_1= 0 ,\\
	 	 	 & ~ j_{\mu} = \left \lceil \frac{-i-s_{\mu}}{m} \right \rceil  ~~\text{for} ~~ \mu =2,\cdots, r,\nonumber \\
	 		&ri+m(j_2+\cdots+j_r)  \leqslant t \Big\} .
	 	\end{align*}
	 It follows immediately that the set $ \Phi $ is not empty if and only if
	 $-rs_1+m \sum\limits_{\mu=2 }^r\left\lceil \dfrac{s_1-s_{\mu}}{m}\right\rceil \leqslant t  $, which concludes the first assertion.

	 It follows from  Corollary~\ref{cor:thetabase2} that
	 the difference between $ \#	\Omega_{s_1, \cdots,s_r, t}$ and $ \#	 \Omega_{s_1,\cdots,s_r, t-1} $ is exactly
	 the same as the one between $ \#\Theta_{s_1,\cdots,s_r, t}$ and $\# \Theta_{s_1,\cdots,s_r, t-1}$.
	 Similar to the argument of the first assertion, we define $ \Psi $ as the
	 complementary set of $\Omega_{s_1,s_2, \cdots,s_r, t-1}$ in $\Omega_{s_1,s_2, \cdots,s_r, t}$, namely
	 			\begin{align*}
	 			\Psi:=\Big\{& (u,v_2,v_3,\cdots,v_r)~	\Big|~ u   = - t \\
	 			& ~ - (a+bm)u-m(v_2+\cdots+v_r) +s_1 \geqslant 0  , \nonumber \\
	 			&~ 0 \leqslant - au + mv_{\mu} +s_{\mu} < m ~~\text{for} ~~ \mu =2,\cdots, r	 			
	 			\Big\}.
	 			\end{align*}
 The set $ \Psi $ is not empty if and only if
$  m \sum\limits_{\mu=2 }^{r}\left\lceil \dfrac{-at-s_{\mu}}{m}\right\rceil
\leqslant s_1 + (a+bm)t $, which finishes the proof of the second assertion.	
\end{proof}

We are now ready for the main results of the section  dealing with the Weierstrass semigroups and the pure gap sets, which play an interesting role in finding codes with good parameters.
\begin{thm}
	\label{lem:Weierstsemi2}
	Let $ P_1,\cdots,P_l $ be the rational places defined previously. For $ 0 \leqslant l \leqslant r $, the following assertions hold.
	\begin{enumerate}
		\item $ H(P_1,\cdots,P_l)$ is given by
		\begin{align*}
		&\Big\{(s_1,\cdots,s_l)\in \mathbb{N}_0^l~\Big |
		~ m \sum\limits_{i=1 \atop i\neq j}^l\left\lceil \dfrac{s_j-s_i}{m}\right\rceil \\
		&+ m (r-l) \left\lceil \dfrac{s_j}{m}\right\rceil
		\leqslant rs_j
		 ~for~ all ~j,  1\leqslant j \leqslant l \Big\}.
		\end{align*}
		\item	
		$	 H(P_1,\cdots,P_{l},P_{\infty}) $ is given by
		 \begin{align*}
		 &\Big \{(s_1,\cdots,s_{l},t)\in \mathbb{N}_0^{l+1}~\Big |
		  ~ m \sum\limits_{i=2 }^{l}\left\lceil \dfrac{-at-s_i}{m}\right\rceil \\
		 & 	+ m (r-l) \left\lceil \dfrac{-at}{m}\right\rceil ~ \leqslant s_1 + (a+bm)t, \\
		 &~ m \sum\limits_{i=1 \atop i\neq j}^{l}\left\lceil \dfrac{s_j-s_i}{m}\right\rceil
			+ m (r-l) \left\lceil \dfrac{s_j}{m}\right\rceil
		 	\leqslant t+ rs_j, \\
		 &~for~ all ~j,  1\leqslant j \leqslant l
		  \Big\} .
		\end{align*}
		\item
	 The pure gap set $ G_0(P_1,\cdots,P_l)$ is given by
		\begin{align*}
		&\Big\{(s_1,\cdots,s_l)\in \mathbb{N}^l~\Big |
		~ m \sum\limits_{i=1 \atop i\neq j}^l\left\lceil \dfrac{s_j-s_i}{m}\right\rceil \\
		&+ m (r-l) \left\lceil \dfrac{s_j}{m}\right\rceil
		> rs_j
		~for~ all ~j,  1\leqslant j \leqslant l \Big\}.
		\end{align*}
				\item	
				The pure gap set $	G_0(P_1,\cdots,P_{l},P_{\infty}) $ is given by
				\begin{align*}
				&\Big \{(s_1,\cdots,s_{l},t)\in \mathbb{N}^{l}~\Big |
				~ m \sum\limits_{i=2 }^{l}\left\lceil \dfrac{-at-s_i}{m}\right\rceil \\
				& 	+ m (r-l) \left\lceil \dfrac{-at}{m}\right\rceil ~ > s_1 + (a+bm)t, \\
				&~ m \sum\limits_{i=1 \atop i\neq j}^{l}\left\lceil \dfrac{s_j-s_i}{m}\right\rceil
				+ m (r-l) \left\lceil \dfrac{s_j}{m}\right\rceil
				> t+ rs_j, \\
				&~for~ all ~j,  1\leqslant j \leqslant l
				\Big\} .
				\end{align*}
	\end{enumerate}
\end{thm}
\begin{proof}
	The desired conclusions follow from Theorem~\ref{thm:basis1},  Lemmas~\ref{lem:Weierstsemi},~\ref{lem:omeganoorder} and~\ref{lem:omegawith1}.
\end{proof}

In recent preprints by Abd\'{o}n \textit{et.al.} \cite{Masuda2,Miriam2015}, the authors determined the Weierstrass semigroup at totally ramified places. Here, we will give another proof by applying our previous results.
\begin{cor}[\cite{Masuda2,Miriam2015}]
	With notation as before, we have the following.
	\begin{enumerate}
 \item
$ H(P_1)= \Big\{\alpha \in \mathbb{N}_0 ~\Big|~-r\alpha+m(r-1) \left \lceil \dfrac{\alpha}{m} \right \rceil \leqslant 0~\Big \} $.\\
  \item
  $ H(P_{\infty})=\Big \{ mk+rj \in \mathbb{N}_0 ~\Big |~0 \leqslant j\leqslant m-1, k \geqslant 0\Big\} $.
  \item
 $G(P_1)$ is given by
  \begin{align*}
  & \Big\{ mk+j \in \mathbb{N}  ~\Big|~1\leqslant j \leqslant m-1-\left\lfloor \dfrac{m}{r} \right\rfloor,\\
  &~0 \leqslant k \leqslant r-2- \left \lfloor \dfrac{rj}{m}
  \right \rfloor ~\Big\}.
  \end{align*}

\item
$ G(P_{\infty})$ is given by
	\begin{align*}
   \Big\{&mk-rj \in \mathbb{N} ~\Big|~
	1\leqslant j \leqslant m-1-\left\lfloor \dfrac{m}{r} \right\rfloor,\\
	&~ \left\lceil\dfrac {rj}{m} \right\rceil \leqslant k \leqslant r-1
	\Big \}.
	\end{align*}

\end{enumerate}	
\end{cor}

\begin{proof}
	The first assertion follows from Theorem~\ref{lem:Weierstsemi2}.

We now focus on the second assertion. It follows from Theorem~\ref{lem:Weierstsemi2} that the Weierstrass semigroup $ H(P_{\infty}) $ can be expressed by
	\begin{equation}\label{eq:HPinf}
	  \Big\{t \in \mathbb{N}_0 ~\Big|~
t-(r-1)at-m(r-1) \left\lceil\dfrac{-at}{m}\right\rceil \geqslant 0
\Big \} ,
  \end{equation}
 where $ ar+ bm=1 $ for integers $a$ and $b$ since $ \gcd(r,m)=1 $. We write the element $t$ in $ H(P_{\infty}) $ as $t=mk+rj$, where
 $0\leqslant j \leqslant m-1$. Then
 $ \left\lceil\dfrac{-at}{m}\right\rceil = bj-ak $. Substituting this to Equation~\eqref{eq:HPinf}, we can deduce that
 	\begin{align*}
 	H(P_{\infty})&= \Big\{mk+rj \in \mathbb{N}_0 ~\Big|~
 	0\leqslant j \leqslant m-1,~~ j+km\geqslant 0
 	\Big \}\\
 	&= \Big\{mk+rj \in \mathbb{N}_0 ~\Big|~
 	0\leqslant j \leqslant m-1, ~~k\geqslant 0
 	\Big \},
 	\end{align*} which concludes the second assertion.
 	
 	It follows from the third assertion of Theorem~\ref{lem:Weierstsemi2} that the Weierstrass gap set at $P_1$ is
 	\begin{equation}\label{eq:GP1}
 	G(P_1)=\Big\{\alpha \in \mathbb{N} ~\Big|~-r\alpha+m(r-1) \left \lceil \dfrac{\alpha}{m} \right \rceil > 0~\Big \} .
 	\end{equation}
  Let $ \alpha \in G(P_1) $ and write $ \alpha=mk+j $, where $ 0 \leqslant j \leqslant m-1 $ and $ k \geqslant 0 $.
 	We find that the case $ j=0 $ does not occur, since otherwise $ -r\alpha+m(r-1) \left \lceil \dfrac{\alpha}{m} \right \rceil = -rkm+m(r-1)k=- \alpha < 0$, which contradict to $ \alpha \in \mathbb{N} $. So $ 1 \leqslant j \leqslant m-1 $.
 	
 	Substituting $ \alpha=mk+j $ to the condition of Equation~\eqref{eq:GP1}, we obtain that
 	\begin{align*}
 	&-r\alpha+m(r-1) \left \lceil \dfrac{\alpha}{m} \right \rceil \\
 	& =-r(mk+j)+m(r-1)(k+1)\\
 	& = -rj-mk+mr-m >0.
 	\end{align*}
 	This yields that $ 0\leqslant k \leqslant \left \lfloor r-1-\dfrac{rj}{m} \right \rfloor =r-2- \left \lfloor \dfrac{rj}{m} \right \rfloor  $, since $ m $ does not divide $ rj $.
 	For the existence of $ k $, we require $  r-1-\dfrac{rj}{m} \geqslant 0 $, which leads to $ 1\leqslant j \leqslant m-1- \left \lfloor \dfrac{m}{r} \right \rfloor$. Thus
 	the Weierstrass gap set $ G(P_1) $ equals
 	\begin{align*}
 	\Big\{& mk+j \in \mathbb{N}~\Big|~1\leqslant j \leqslant m-1-\left\lfloor \dfrac{m}{r} \right\rfloor,\\
 	&~0 \leqslant k \leqslant r-2- \left \lfloor \dfrac{rj}{m}
 	\right \rfloor ~\Big\},
 	\end{align*}
  	 and we get the desired conclusion.

 	 From Equation~\eqref{eq:HPinf}, we have
 	 \begin{equation*}
 	  G(P_{\infty})= \Big\{t \in \mathbb{N} ~\Big|~
 	 t-(r-1)at-m(r-1) \left\lceil\dfrac{-at}{m}\right\rceil < 0
 	 \Big \} .
 	 \end{equation*}
 	 Write $t=mk-rj$, where $0\leqslant j \leqslant m-1$. In a similar way as done previously for the third assertion, we see that $1\leqslant j \leqslant m-1$. Therefore, substituting $t=mk-rj$ to the inequality of $ G(P_{\infty}) $, we have $ k\leqslant r-1 $.
 	 Since $ mk-rj \in \mathbb{N} $, $ 1\leqslant j \leqslant m-1$ and
 	 $ k \leqslant r-1 $, we obtain $ \left\lceil\dfrac {rj}{m} \right\rceil \leqslant k \leqslant r-1 $. Also note that $ \dfrac {rj}{m}  < r-1  $, which yields that $ 1\leqslant j \leqslant m-1-\left\lfloor \dfrac{m}{r} \right\rfloor $. This completes the proof of the last assertion.
\end{proof}

The following is an immediate consequence of Theorem~\ref{lem:Weierstsemi2} which generalizes the results of Theorems 3.2 and 4.4 in~\cite{Masuda2}.
\begin{cor} \label{cor:semigrouppuregap}The following assertions hold.
	\begin{enumerate}
		\item
		$ H(P_1,P_2) $ is given by
		\begin{align*}
		&\Big\{ (k,l)\in \mathbb{N}_0^2 ~	\Big	|	
		~ m\left \lceil\dfrac{k-l}{m}\right\rceil + m(r-2)\left \lceil\dfrac{k}{m}\right\rceil \leqslant kr, \\
		&~~ m\left \lceil\dfrac{l-k}{m}\right\rceil + m(r-2)\left \lceil\dfrac{l}{m}\right\rceil \leqslant lr  \Big\}.
		\end{align*}
		\item
		$ G_0(P_1,P_2) $ is given by
		\begin{align*}
		& \Big\{ (k,l)\in \mathbb{N}^2 ~	\Big	|	
		~ m\left \lceil\dfrac{k-l}{m}\right\rceil + m(r-2)\left \lceil\dfrac{k}{m}\right\rceil > kr, \\
		&~ m\left \lceil\dfrac{l-k}{m}\right\rceil + m(r-2)\left \lceil\dfrac{l}{m}\right\rceil > lr  \Big\}.
		\end{align*}
		\item
		$ H(P_1,P_{\infty}) $ is given by
		\begin{align*}
		&\Big\{ (k,t)\in \mathbb{N}_0^2 ~	\Big	|	
		~ m(r-1)\left \lceil\dfrac{k}{m}\right\rceil  \leqslant t+kr, \\
		&~~ m(r-1)\left \lceil\dfrac{-at}{m}\right\rceil  \leqslant k+(a+bm)t  \Big\}.
		\end{align*}
		\item
		$ G_0(P_1,P_{\infty}) $ is given by
		\begin{align*}
		&\Big\{ (k,t)\in \mathbb{N}^2 ~	\Big	|	
		~ m(r-1)\left \lceil\dfrac{k}{m}\right\rceil  > t+kr, \\
		&~~ m(r-1)\left \lceil\dfrac{-at}{m}\right\rceil  > k+(a+bm)t  \Big\}.
		\end{align*}
		                     	
	\end{enumerate}
\end{cor}

\section{The floor of divisors }

In this section, we investigate the floor of divisors of function fields. The
significance of this concept is that it provides a useful
tool for evaluating parameters of AG codes, as well as pure gaps. We begin with
general function fields.
\begin{defn}[\cite{Maharaj2005riemann}]
	Given a divisor $ G $ of a function field $ F/\mathbb{F}_q $ with $ \ell(G)>0 $, the floor of $ G $ is the unique divisor $ G' $ of $ F $ of minimum degree such that $ \mathcal{L}(G)=\mathcal{L}(G') $. The floor of $ G $ will be denoted by $ \lfloor G \rfloor $.
\end{defn}
The floor of a divisor can be used to characterize Weierstrass semigroups and pure gap sets. Let $ G=s_1Q_1 + \cdots+ s_lQ_l $. It is not hard to see that $ (s_1,\cdots,s_l) \in H (Q_1,\cdots,Q_l)$ if and only if $ \lfloor G \rfloor=G $. Moreover,  $ (s_1,\cdots,s_l) $ is a pure gap at $ (Q_1,\cdots,Q_l) $ if and only if
\begin{equation*}
\lfloor G \rfloor= \lfloor (s_1-1)Q_1 + \cdots+ (s_l-1)Q_l \rfloor.
\end{equation*}

Maharaj, Matthews and Pirsic in~\cite{Maharaj2005riemann} defined the floor of a divisor and characterized it by the basis of the Riemann-Roch space.
\begin{thm}[\cite{Maharaj2005riemann}]\label{thm:floorofG}
	Let $ G $ be a divisor of the function field
	$ F/\mathbb{F}_q $ and let $ b_1, \cdots, b_t \in  \mathcal{L}(G)$ be a spanning set for $ \mathcal{L}(G)$. Then
	\begin{equation*}
	\lfloor G \rfloor=-\gcd\Big\{(b_i)~\Big|~i=1,\cdots,t\Big\}.
	\end{equation*}
\end{thm}
The next theorem extends Theorem~\ref{thm:Goppagener}, which shows the lower bound of minimal distance in a more general situation.

\begin{thm}[\cite{Maharaj2005riemann}]\label{thm:Maharajfloor}
	Let $ F/\mathbb{F}_q $ be a function field of genus $ g $.
	Let $ D:=Q_1+\cdots+Q_n $ where $ Q_1,\cdots, Q_n $ are distinct rational places of $ F $, and let $ G:= H+\lfloor H \rfloor $ be a divisor of $ F $
	such that $ H $ is an effective divisor whose support does not contain any of the places $ Q_1,\cdots, Q_n $. Then the distance of $ C_{\Omega} $ satisfies
	\begin{align*}
	d_{\Omega} \geqslant 2 \deg(H)-(2g-2).
	\end{align*}
\end{thm}

To finish this section, we provide a characterization of the floor over Kummer extensions. The following theorem is a generalization of Theorem 3.9 in~\cite{Maharaj2005riemann} related to Hermitian function fields.
\begin{thm}\label{thm:floorofH}
	Let $ H :=s_1P_1+s_2P_2+\cdots+s_rP_r+tP_{\infty} $ be a divisor of the Kummer extension given by~\eqref{eq:Kumext1}. Then the floor of $ H $ is given by
	\begin{align*}
	\lfloor H \rfloor = s_1'P_1+s_2'P_2+\cdots+s_r'P_r+t'P_{\infty},
	\end{align*}
	where
	\begin{align*}
	s_1'&=\max\Big\{-i ~\Big|~ (i,j_2,\cdots,j_r)\in \Omega_{s_1,\cdots,s_r,t}\Big\},\\
	s_{\mu}'&= \max\Big\{ -i-mj_{\mu}~\Big|~(i,j_2,\cdots,j_r)\in \Omega_{s_1,\cdots,s_r,t}\Big\},\\
	&\hspace{4cm}
	 \text{for} ~\mu = 2,\cdots, r,\\
	t'& = \max\Big\{ri+m\sum_{\mu=2}^{r}j_{\mu}~\Big|~(i,j_2,\cdots,j_r)\in \Omega_{s_1,\cdots,s_r,t}\Big\}.
	\end{align*}
\end{thm}
\begin{proof}
		Let $ H =s_1P_1+\cdots+s_rP_r+tP_{\infty} $. It follows from Theorem~\ref{thm:basis1} that the elements $ E_{i,j_2,j_3,\cdots,j_r} $ of Equation~\eqref{eq:Eijr} with $ (i,j_2,\cdots,j_r) \in \Omega_{s_1,\cdots,s_r, t} $ form a basis for the Riemann-Roch space
	$ \mathcal{L}(H) $.
	Note that the divisor of $ E_{i,j_2,\cdots,j_r} $  is
	\begin{equation*}
	iP_1 + \sum_{\mu=2}^{r}(i+m j_{\mu}) P_{\mu} -\left( ri+m\sum_{\mu=2}^{r} j_{\mu} \right) P_{\infty}.
	\end{equation*}	
	By Theorem~\ref{thm:floorofG}, we get that
	\begin{equation*}
	\lfloor H \rfloor =-\gcd\Big\{(E_{i,j_2,\cdots,j_r})~\Big|~ (i,j_2,\cdots,j_r) \in \Omega_{s_1,\cdots,s_r, t} \Big\}.
	\end{equation*}	
	The desired conclusion then follows.
\end{proof}

\section{Examples of codes in Kummer extensions}\label{sec:Examples}

 In this section we treat several examples of codes to illustrate our results. All the codes in our examples have better parameters than the corresponding ones in the MinT's tables~\cite{MinT}.

\begin{example}
	Now, we study codes arising from plane quotients of the Hermitian curve, defined by affine equations of the form $ y^m=x^q+x $ over $ \mathbb{F}_{q^2} $, where $ q $ is a prime power and $m$ is a positive integer which divides $q+1 $. Take $ q=9 $ and $ m=5 $ for example. It follows from~\cite{ballico2013dual} that the number of rational places of the curve $ y^5=x^9+x $ with genus $ g=16 $ is $ N=1+q(1+(q-1)m)=370 $. From Corollary~\ref{cor:semigrouppuregap}, we can get all the pure gaps at $ (P_1,P_{\infty}) $, which are showed in Figure~\ref{fig:puregaps}. Choose $n$ (for example $ 360 \leqslant n \leqslant 368 $) rational places with the exceptions of $ P_1 $ and $ P_{\infty} $ and consider the divisor $ D $ as the sum of these places. We choose a pure gap $ (26,1) $ for instance. It follows from Theorem~\ref{thm:Goppagener} that, if we take $ G=51P_1+P_{\infty} $, then the minimum distance of $ C_{\Omega} $ satisfies
	$ d_{\Omega} \geqslant 24 $.
	One can easily check that the condition $ 2g-2 < \deg(G) < n $ holds, so Equation~\eqref{eq:dimofCLk} yields that the dimension of $ C_{\Omega} $ is $ k_{\Omega}=n-37 $.
	Thus we obtain a class of codes with parameters $ [n,n-37,{\geqslant 24}] $. The minimum distance of each code exceeds the minimum distance of the best known codes over $ \mathbb{F}_{81} $ with the same length and dimension in the MinT's Tables~\cite{MinT}.
\end{example}
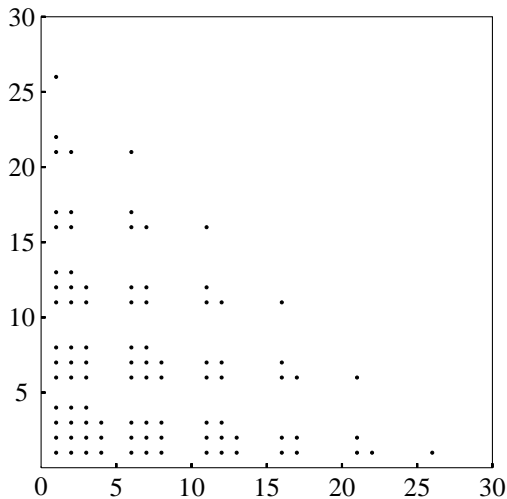
\begin{figure}[H]
	\centering
	\begin{tikzpicture}[scale=0.2]
	\draw [black] (0,0)--(30,0)--(30,30)--(0,30)--(0,0);

    \draw  node[below]{0} -- (0,0.1);
	\draw [thick] (5,0)--(5,.3)   node[yshift=-2ex,xshift=0ex]{5} -- (5,0);
    \draw [thick] (10,0)--(10,.3) node[yshift=-2ex,xshift=0ex]{10} -- (10,0);
	\draw [thick] (15,0)--(15,.3) node[yshift=-2ex,xshift=0ex]{15} -- (15,0);
	\draw [thick] (20,0)--(20,.3) node[yshift=-2ex,xshift=0ex]{20} -- (20,0);
	\draw [thick] (25,0)--(25,.3) node[yshift=-2ex,xshift=0ex]{25} -- (25,0);
	\draw [thick] (30,0)--(30,.3) node[yshift=-2ex,xshift=0ex]{30} -- (30,0);

	\draw [thick] (0,5)--(0.3,5)   node[yshift=0ex,xshift=-2ex]{5} -- (0,5);
    \draw [thick] (0,10)--(0.3,10) node[yshift=0ex,xshift=-2ex]{10} -- (0,10);
    \draw [thick] (0,15)--(0.3,15) node[yshift=0ex,xshift=-2ex]{15} -- (0,15);
    \draw [thick] (0,20)--(0.3,20) node[yshift=0ex,xshift=-2ex]{20} -- (0,20);
    \draw [thick] (0,25)--(0.3,25) node[yshift=0ex,xshift=-2ex]{25} -- (0,25);
    \draw [thick] (0,30)--(0.3,30) node[yshift=0ex,xshift=-2ex]{30} -- (0,30);

\draw [fill] (1, 1) circle [radius=0.1];
\draw [fill] (1, 2) circle [radius=0.1];
\draw [fill] (1, 3) circle [radius=0.1];
\draw [fill] (1, 4) circle [radius=0.1];
\draw [fill] (1, 6) circle [radius=0.1];
\draw [fill] (1, 7) circle [radius=0.1];
\draw [fill] (1, 8) circle [radius=0.1];
\draw [fill] (1, 11) circle [radius=0.1];
\draw [fill] (1, 12) circle [radius=0.1];
\draw [fill] (1, 13) circle [radius=0.1];
\draw [fill] (1, 16) circle [radius=0.1];
\draw [fill] (1, 17) circle [radius=0.1];
\draw [fill] (1, 21) circle [radius=0.1];
\draw [fill] (1, 22) circle [radius=0.1];
\draw [fill] (1, 26) circle [radius=0.1];
\draw [fill] (2, 1) circle [radius=0.1];
\draw [fill] (2, 2) circle [radius=0.1];
\draw [fill] (2, 3) circle [radius=0.1];
\draw [fill] (2, 4) circle [radius=0.1];
\draw [fill] (2, 6) circle [radius=0.1];
\draw [fill] (2, 7) circle [radius=0.1];
\draw [fill] (2, 8) circle [radius=0.1];
\draw [fill] (2, 11) circle [radius=0.1];
\draw [fill] (2, 12) circle [radius=0.1];
\draw [fill] (2, 13) circle [radius=0.1];
\draw [fill] (2, 16) circle [radius=0.1];
\draw [fill] (2, 17) circle [radius=0.1];
\draw [fill] (2, 21) circle [radius=0.1];
\draw [fill] (3, 1) circle [radius=0.1];
\draw [fill] (3, 2) circle [radius=0.1];
\draw [fill] (3, 3) circle [radius=0.1];
\draw [fill] (3, 4) circle [radius=0.1];
\draw [fill] (3, 6) circle [radius=0.1];
\draw [fill] (3, 7) circle [radius=0.1];
\draw [fill] (3, 8) circle [radius=0.1];
\draw [fill] (3, 11) circle [radius=0.1];
\draw [fill] (3, 12) circle [radius=0.1];
\draw [fill] (4, 1) circle [radius=0.1];
\draw [fill] (4, 2) circle [radius=0.1];
\draw [fill] (4, 3) circle [radius=0.1];
\draw [fill] (6, 1) circle [radius=0.1];
\draw [fill] (6, 2) circle [radius=0.1];
\draw [fill] (6, 3) circle [radius=0.1];
\draw [fill] (6, 6) circle [radius=0.1];
\draw [fill] (6, 7) circle [radius=0.1];
\draw [fill] (6, 8) circle [radius=0.1];
\draw [fill] (6, 11) circle [radius=0.1];
\draw [fill] (6, 12) circle [radius=0.1];
\draw [fill] (6, 16) circle [radius=0.1];
\draw [fill] (6, 17) circle [radius=0.1];
\draw [fill] (6, 21) circle [radius=0.1];
\draw [fill] (7, 1) circle [radius=0.1];
\draw [fill] (7, 2) circle [radius=0.1];
\draw [fill] (7, 3) circle [radius=0.1];
\draw [fill] (7, 6) circle [radius=0.1];
\draw [fill] (7, 7) circle [radius=0.1];
\draw [fill] (7, 8) circle [radius=0.1];
\draw [fill] (7, 11) circle [radius=0.1];
\draw [fill] (7, 12) circle [radius=0.1];
\draw [fill] (7, 16) circle [radius=0.1];
\draw [fill] (8, 1) circle [radius=0.1];
\draw [fill] (8, 2) circle [radius=0.1];
\draw [fill] (8, 3) circle [radius=0.1];
\draw [fill] (8, 6) circle [radius=0.1];
\draw [fill] (8, 7) circle [radius=0.1];
\draw [fill] (11, 1) circle [radius=0.1];
\draw [fill] (11, 2) circle [radius=0.1];
\draw [fill] (11, 3) circle [radius=0.1];
\draw [fill] (11, 6) circle [radius=0.1];
\draw [fill] (11, 7) circle [radius=0.1];
\draw [fill] (11, 11) circle [radius=0.1];
\draw [fill] (11, 12) circle [radius=0.1];
\draw [fill] (11, 16) circle [radius=0.1];
\draw [fill] (12, 1) circle [radius=0.1];
\draw [fill] (12, 2) circle [radius=0.1];
\draw [fill] (12, 3) circle [radius=0.1];
\draw [fill] (12, 6) circle [radius=0.1];
\draw [fill] (12, 7) circle [radius=0.1];
\draw [fill] (12, 11) circle [radius=0.1];
\draw [fill] (13, 1) circle [radius=0.1];
\draw [fill] (13, 2) circle [radius=0.1];
\draw [fill] (16, 1) circle [radius=0.1];
\draw [fill] (16, 2) circle [radius=0.1];
\draw [fill] (16, 6) circle [radius=0.1];
\draw [fill] (16, 7) circle [radius=0.1];
\draw [fill] (16, 11) circle [radius=0.1];
\draw [fill] (17, 1) circle [radius=0.1];
\draw [fill] (17, 2) circle [radius=0.1];
\draw [fill] (17, 6) circle [radius=0.1];
\draw [fill] (21, 1) circle [radius=0.1];
\draw [fill] (21, 2) circle [radius=0.1];
\draw [fill] (21, 6) circle [radius=0.1];
\draw [fill] (22, 1) circle [radius=0.1];
\draw [fill] (26, 1) circle [radius=0.1];

	\end{tikzpicture}
	
	\protect\caption{The pure gaps at $ (P_1,P_{\infty}) $ }
	\label{fig:puregaps}
\end{figure}
\begin{example}
	The Hermitian curve $ y^{q+1}=x^q+x $ over $ \mathbb{F}_{q^2} $ has $ q^3+1 $ rational places and genus $ g=q(q-1)/2 $. If we take $ q=5 $, then $ g=10 $. From Corollary~\ref{cor:semigrouppuregap}, one can verify that $ (13,1) $ and $ (14, 1) $ are pure gaps at $ (P_1, P_2) $. Choose $ n $ arbitrary rational places on the curve except $ P_1$ and $ P_2 $, where $ 120\leqslant n \leqslant 124 $. By Theorem~\ref{thm:Goppagener}, if we take $ G=26P_1+P_2 $, then the minimum distance of $ C_{\Omega} $ satisfies
	$ d_{\Omega} \geqslant 12 $. Since $ 2g-2<\deg(G)<n $, the dimension of such code is $ k_{\Omega}=n+g-1-\deg(G)=n-18 $. This is to say that we obtain a class of codes with parameters $ [n,n-18,{\geqslant 12}] $. Based on MinT's Tables, these codes improve the minimum distance.
\end{example}

\begin{example}
	Consider the curve $ y^{m}=(x^{q^{t/2}}-x)^{q^{t/2}-1} $ over
		$ \mathbb{F}_{q^{t}} $, where $ t $ is even, $ q $ is a prime power,
		$ m ~\big|~(q^{t}-1) $ and $ \gcd(m,q^{t/2}-1)=1 $. From~\cite{Garcia2001}, we know that this curve has genus $ g=(q^{t/2}-1)(m-1)/2 $ and the number of
		$ \mathbb{F}_{q^{t}} $-rational places is
		$ N=(q^{t} - q^{{t}/2})m+q^{t/2}+1 $.	
		 Let $ t=2 $, $ q=5 $ and $ m=6 $. The curve becomes $ y^{6}=(x^{5}-x)^{4} $ with $ g=10 $ and $ N=126 $ over
		 $ \mathbb{F}_{25} $. By Theorem~\ref{lem:Weierstsemi2}, one can verify that
		 \begin{align*}
		 \Big\{(i,j,k)~\Big|~8\leqslant i \leqslant 9, j=1, 1\leqslant k \leqslant 3\Big\} \subseteq G_0(P_1,P_2,P_{\infty}).
		 \end{align*}
		 Let $ D $ be the divisor consisting of all rational places away from places $ P_1,P_2$ and $P_{\infty} $. Taking $ G=16P_1+P_2+3P_{\infty} $, we produce the codes $ C_{\Omega} $ of length $ n=\deg(D)=123 $. It follows from Theorem~\ref{thm:Goppagener},
		 that the minimum distance of $ C_{\Omega} $ satisfies
		 $ d_{\Omega} \geqslant 8 $. Since $ 2g-2<\deg(G)< n $, the dimension of $ C_{\Omega} $ is $ k_{\Omega}=n+g-1-\deg(G)=n-11 $.
		 So the three-point code $ C_{\Omega} $ has parameters $ [n,n-11,\geqslant 8] $, which improves the minimum distance according to MinT's Tables.
\end{example}

\begin{example}
	Let us consider the maximal curve~\cite{abdon1999maximal} $ y^{q+1}= \sum_{i=1}^t x^{q/{2^i} }$ over $ \mathbb{F}_{q^2} $, with genus $ g=q(q-2)/4 $, where $ q=2^t $. Taking $t=3$, the curve becomes $ y^9=x^4+x^2+x  $ over $\mathbb{F}_{64} $ with $ N=257 $ rational places. Let $ H = \gamma P_1+P_2+4P_{\infty} $, where $ \gamma $ is an integer such that $ 14 \leqslant \gamma \leqslant 18 $. We put $ \gamma=14 $ for instance. In this case $ H = 14P_1+P_2+4P_{\infty} $, it can be computed from Equation~\eqref{eq:omegarceiljmu} that the elements $ (-i,-i-mj_2,-i-m j_3,-i-mj_4,ri+m(j_2+j_3+j_4)) $, with $(i,j_2,j_3,j_4)  \in  \Omega_{14,1,0,0,4} $, are as follows
	\begin{align*}
	  &( \,14,  -4,    -4 ,   -4 ,    -2~),\\
	  &( \,13,    -5 ,    -5,    -5 ,    \phantom{-}2~),\\
	  &( \,\phantom{0}9,     \phantom{-}0 ,   \phantom{-} 0 ,  \phantom{-}  0 ,    -9~),\\
	  &(\,\phantom{0}8,     -1 ,    -1 ,    -1,    -5~),\\
	  &(\,\phantom{0}7,     -2,     -2 ,    -2 ,    -1~),\\
	  &(\,\phantom{0}6,     -3 ,   -3 ,    -3 ,\phantom{-}   3~),\\
	  &(\, \phantom{0}0,     \phantom{-}0 ,    \phantom{-} 0 ,\phantom{-}    0 ,\phantom{-}   0~),\\
	  &( -1,     -1,     -1 ,    -1,\phantom{-}  4~).
	\end{align*}
	
	We obtain from Theorem~\ref{thm:floorofH} that $ \lfloor H \rfloor = 14P_1 + 4 P_{\infty} $.
	Choose $ n=254 $ rational places with the exceptions of $ P_1,P_2$ and $P_{\infty} $.
	According to Theorem~\ref{thm:Maharajfloor}, if we let
	$ G=H+\lfloor H \rfloor= 28 P_1 + P_2 +8 P_{\infty} $, then the code
	$ C_{\Omega} $ is a $ [254,228,\geqslant 16] $ code. For $ 14 \leqslant \gamma \leqslant 18 $, one can verify that all of the resulting codes with parameters $ [254,256-2\gamma,\geqslant 2\gamma-12] $ improve the minimum distance with respect to MinT's Tables.
\end{example}


%





\ifCLASSOPTIONcaptionsoff
  \newpage
\fi

\end{document}